\documentclass{article}
\usepackage{amsthm}
\usepackage{amsmath}
\usepackage{amssymb}
\usepackage{amsfonts}
\usepackage{color}
\begin{document}

\newtheorem{theorem}{Theorem}[section]
\newtheorem{definition}[theorem]{Definition}
\newtheorem{corollary}[theorem]{Corollary}
\newtheorem{lemma}[theorem]{Lemma}
\newtheorem{proposition}[theorem]{Proposition}
\newtheorem{step}[theorem]{Step}
\newtheorem{example}[theorem]{Example}
\newtheorem{remark}[theorem]{Remark}

\font\sixbb=msbm6
\font\eightbb=msbm8
\font\twelvebb=msbm10 scaled 1095

\newfam\bbfam
\textfont\bbfam=\twelvebb \scriptfont\bbfam=\eightbb
                           \scriptscriptfont\bbfam=\sixbb

\def\cP{{\mathcal{P}}}
\def\bC{{\mathbb{C}}}
\def\bN{{\mathbb{N}}}
\def\bR{{\mathbb{R}}}

\title{Nonexpansive and noncontractive mappings on the set of quantum pure states\thanks{The first author was supported by JSPS KAKENHI Grant Number 22K13934. The second author was supported by grants J1-2454 and P1-0288 from ARRS, Slovenia.}}
\author{Michiya Mori\footnote{Graduate School of Mathematical Sciences, The University of Tokyo, 3-8-1 Komaba, Meguro-ku, Tokyo, 153-8914, Japan, mmori@ms.u-tokyo.ac.jp; Interdisciplinary Theoretical and Mathematical Sciences Program (iTHEMS), RIKEN, 2-1 Hirosawa, Wako, Saitama, 351-0198, Japan}  \ \, and \ \, 
Peter \v Semrl\footnote{Institute of Mathematics, Physics, and Mechanics, Jadranska 19, SI-1000 Ljubljana, Slovenia, peter.semrl@fmf.uni-lj.si}
        }

\date{}
\maketitle

\begin{abstract}
Wigner's theorem characterizes isometries of the set of all rank one projections on a Hilbert space. In metric geometry nonexpansive maps and noncontractive maps are well studied generalizations of isometries. We show that under certain conditions Wigner symmetries can be characterized as nonexpansive or noncontractive maps on the set of all projections of rank one. The assumptions required for such characterizations are injectivity or surjectivity and they differ in the finite and the infinite-dimensional case. 
Motivated by a recently obtained optimal version of Uhlhorn's generalization of Wigner's theorem,  
we also give a description of nonexpansive maps which satisfy a condition that is much weaker than surjectivity. 
Such maps do not need to be Wigner symmetries. 
The optimality of all presented results is shown by counterexamples. 
\end{abstract}
\maketitle




\section{Introduction and statement of main results}

Throughout this note $H$ is a complex Hilbert space and $\cP(H)$ denotes the collection of rank one projections on $H$. We will always assume that $\dim H \ge 2$.
In the mathematical foundations of quantum mechanics projections of rank one
represent pure states of a quantum system, and the quantity ${\rm tr}\, (PQ)$, the trace of the product $PQ$, $P,Q \in \cP(H)$, is the so-called transition probability between two pure states. 
The famous Wigner's unitary-antiunitary theorem \cite{Wig} describes the general form of bijective transformations of $\cP(H)$ which preserve the transition probability.

We may state Wigner's theorem in terms of isometry. 
Recall that a mapping $\phi$ from a metric space $(X, d)$ into itself is called an isometry if 
$$d(\phi (x),\phi (y)) = d(x,y), \ \ \ x, y \in X.$$
In this note, we consider the distance $d(P, Q)=\|P-Q\|$ in $\cP(H)$, where $\|\cdot\|$ denotes the operator norm.
One can easily verify that $\|P-Q\| = \sqrt{1 - {\rm tr}\, (PQ)}$ holds true for every pair $P,Q \in \cP(H)$. Hence, Wigner's theorem can be formulated in the following way: Let $\phi : \cP(H) \to \cP (H)$ be a bijective map such that $$\| \phi (P) - \phi (Q) \| = \| P-Q \|, \ \ \ P,Q \in \cP (H).$$ Then there exists a unitary or an antiunitary operator $U : H \to H$ such that $$\phi (P) = UPU^\ast, \ \ \ P \in \cP (H).$$ Every map of this form will be called a Wigner symmetry.
After all, Wigner's theorem says that every surjective isometry from $\cP(H)$ onto itself is a Wigner symmetry.
More generally, it is known as a nonbijective version of Wigner's theorem that every isometry from $\cP(H)$ into itself is of the form $\phi (P) = UPU^\ast$, $P \in \cP (H)$, for some linear or conjugate-linear isometry $U : H \to H$ (see e.g.\ \cite{Ge0}).
For more detailed explanation with many other references and some recent improvements of Wigner's theorem we refer to \cite{FKKS, GeM, GeS, Mo6, Pan, PaV, Sem}.

In this note we will deal with the question what happens if we replace isometries in Wigner's theorem with maps that are nonexpansive or noncontractive. 
Recall that a mapping $\phi$ from a metric space $(X, d)$ to itself is said to be nonexpansive (resp.\ noncontractive) if 
$$d(\phi (x),\phi (y)) \leq d(x,y) \ \ \ (\text{resp.\, }d(\phi (x),\phi (y)) \geq d(x,y)), \ \ \ x, y \in X.$$
The following is a well-known fact from metric geometry \cite[Theorems 1.6.14 and 1.6.15]{BBI}.
\begin{theorem}\label{metric}
Let $(X, d)$ be a compact metric space. 
\begin{enumerate}
\item A noncontractive map from $X$ into itself is a surjective isometry. 
\item A surjective nonexpansive map from $X$ onto itself is an isometry. 
\end{enumerate}
\end{theorem} 

\begin{corollary}\label{known}
Assume that $\dim H<\infty$. A mapping $\phi\colon \cP(H)\to \cP(H)$ is a Wigner symmetry if one of the following holds.
\begin{enumerate}
\item $\phi$ is noncontractive. 
\item $\phi$ is surjective and nonexpansive. 
\item $\phi$ is injective and nonexpansive. 
\end{enumerate}
\end{corollary}
\begin{proof}
If $\dim H<\infty$, then ${\cal P}(H)$ is compact. 
Therefore, if $\phi$ satisfies one of the first two conditions, then Theorem \ref{metric} and Wigner's theorem imply that $\phi$ is a Wigner symmetry. 
To consider the third condition, let $\phi\colon \cP(H)\to \cP(H)$ be an injective nonexpansive map. 
Then $\phi$ is an injective continuous map from a compact connected manifold ${\cal P}(H)$ into itself. 
By the invariance of domain its image is open. 
By the continuity, its image is compact. Hence, $\phi$ is surjective 
and satisfies the second condition, so it is a Wigner symmetry. 
\end{proof}
We study nonexpansive or noncontractive maps of $\cP(H)$ under various additional assumptions, including the case of $\dim H=\infty$.
We start with a result on noncontractive maps.

\begin{theorem}\label{plasticity3}
Let $\phi : {\cal P} (H) \to {\cal P} (H)$ be a surjective map such that 
$$
\| \phi (P) - \phi (Q) \| \ge \| P - Q \|
$$
for every pair $P,Q \in {\cal P}(H)$. Then $\phi$ is a Wigner symmetry. 
\end{theorem}
Looking at the first item of Corollary \ref{known}, in the finite-dimensional case we see that the same conclusion holds without the surjectivity assumption. However, in the infinite-dimensional case the surjectivity assumption is essential (Example \ref{ce1}).

We next consider nonexpansive maps.
\begin{theorem}\label{plasticity}
Let 
$\phi : {\cal P} (H) \to {\cal P} (H)$ be a surjective map such that 
$$
\| \phi (P) - \phi (Q) \| \le \| P - Q \|
$$
for every pair $P,Q \in {\cal P}(H)$. Then $\phi$ is a Wigner symmetry.
\end{theorem}
In contrast, we will see in Example \ref{ef} that an injective nonexpansive map of $\cP(H)$ can be far from a Wigner symmetry if $\dim H=\infty$. 

For $P,Q \in {\cal P}(H)$ we write $P \perp Q$ if $PQ=0$ which is equivalent to $\| P - Q\| = 1$. For every vector $x \in H$ of norm one we denote by $P_x$ the rank one projection onto the linear span of $x$. We will say that a subset $\{ P_\alpha \, : \, \alpha \in {\cal J} \} \subset {\cal P}(H)$ is  an  
orthogonal system of projections of rank one, OSP,
if $P_\alpha \perp P_\beta$ whenever $\alpha \not=\beta$, and we will say that it is  a complete  
orthogonal system of projections of rank one, COSP,
if it is an OSP and there is no rank one projection $Q$ that is orthogonal to each $P_\alpha$. If $P_\alpha = P_{x_\alpha}$ for some unit vectors $x_\alpha$, $\alpha \in {\cal J}$, then  $\{ P_\alpha \, : \, \alpha \in {\cal J} \} \subset {\cal P}(H)$ is a COSP if and only if $\{ x_\alpha \, : \, \alpha \in {\cal J} \}$ is an orthonormal basis of $H$.

Under the additional assumption $\dim H\geq 3$, Uhlhorn's improvement of Wigner's theorem \cite{Uhl} states that every bijective map $\phi : {\cal P} (H) \to {\cal P} (H)$ with the property that for every pair $P,Q \in \cP (H)$ we have $P \perp Q \iff \phi (P) \perp \phi (Q)$ is a Wigner symmetry. This result can be further improved. A recently obtained optimal version \cite{Sem} reads as follows: Assume that $H$ is separable and $\phi : {\cal P} (H) \to {\cal P} (H)$ a map such that for every pair $P,Q \in \cP (H)$ we have $P \perp Q \Rightarrow \phi (P) \perp \phi (Q)$. Suppose that there exists a COSP $\{ P_1, P_2, \ldots \} \subset \cP (H)$ such that 
$\{ \phi (P_1 ), \phi (P_2 ), \ldots \} \subset \cP (H)$ is a COSP. Then $\phi$ is a Wigner symmetry. In particular, if $3\leq \dim H<\infty$, then every map
$\phi : {\cal P} (H) \to {\cal P} (H)$ with the property that for every pair $P,Q \in \cP (H)$ we have $P \perp Q \Rightarrow \phi (P) \perp \phi (Q)$ is a Wigner symmetry.

Note that every noncontractive map $\phi : {\cal P} (H) \to {\cal P} (H)$ satisfies $P \perp Q \Rightarrow \phi (P) \perp \phi (Q)$, $P,Q \in \cP (H)$. 
It follows that in the separable case a noncontractive map $\phi$ on $\cP (H)$ with the property that there is a COSP that is mapped by $\phi$ onto some COSP must be a Wigner symmetry. The situation becomes much more interesting when we consider nonexpansive maps, as described below.

We first choose and fix an orthonormal basis $\{e_\alpha\,:\, \alpha\in {\cal J}\}$ in $H$. 
For each projection $P\in \cP(H)$ onto $\bC(\sum_{\alpha\in {\cal J}}c_\alpha e_\alpha)$ we define $\Phi(P)$ to be the projection onto $\bC(\sum_{\alpha\in {\cal J}}|c_\alpha|e_\alpha)$. Clearly, $\Phi$ is well-defined, and it maps every element of COSP $\{ P_{e_\alpha}\, : \, \alpha\in {\cal J}\}$ to itself. We claim that $\Phi : \cP (H) \to \cP (H)$ is a nonexpansive map. Equivalently, we need to show that
$$
{\rm tr}\, (\Phi (P)  \Phi (Q))  \ge {\rm tr}\, ( P  Q ), \ \ \ P,Q \in  {\cal P}(H).
$$
This can be easily checked by the inequality 
\[
\sqrt{{\rm tr}\, (P_vP_w)}= \left|\sum_{\alpha \in {\cal J}}v_\alpha\overline{w_\alpha}\right|\leq \sum_{\alpha \in {\cal J}}|v_\alpha||w_\alpha|=\sqrt{{\rm tr}\, (\Phi(P_v)\Phi(P_w))}
\] 
for each pair of unit vectors $v=\sum_{\alpha \in {\cal J}} v_\alpha e_\alpha, w=\sum_{\alpha \in {\cal J}}w_\alpha e_\alpha\in H$.
Now we state the main result of this note. 

\begin{theorem}\label{Phiinfin}
Assume that $\dim H \ge 3$.
Let $\phi\colon \cP(H)\to \cP(H)$ be a mapping satisfying 
\[
\|\phi(P)-\phi(Q)\|\leq \|P-Q\|,\quad P,Q\in \cP(H),
\]
and assume that there exists a COSP $\{ Q_\alpha \, : \, \alpha \in {\cal J} \} \subset \cP (H)$ such that $\{ \phi (Q_\alpha) \, : \, \alpha \in {\cal J} \}$ is a COSP in $\cP (H)$.
Then either $\phi$ is a Wigner symmetry, or there exist unitary operators $U,V : H \to H$ such that
\begin{equation}\label{mmnk}
\phi (P)= V \Phi (UPU^\ast) V^\ast , \ \ \ P \in {\cal P}(H) .
\end{equation}
\end{theorem}

In fact, we will prove a stronger result (Proposition \ref{Phicor}) on an arbitrary nonexpansive map $\phi\colon \cP(H)\to \cP(H)$ whose image contains a COSP.
If $H$ is finite-dimensional, $\dim H = n$, then we identify $H$ with $\mathbb{C}^n$ and ${\cal P}(\mathbb{C}^n)$ with the set of all idempotent hermitian rank one $n \times n$ matrices. 
Then the map $\Phi\colon \cP(\bC^n)\to \cP(\bC^n)$ with respect to the standard orthonormal basis of $\bC^n$ is given by the ``entrywise absolute value'' 
\begin{equation}\label{ab}
\Phi(P)=[\, |p_{ij}|\, ]_{1\leq i,j\leq n} \ \ \text{for}\ \ P=[p_{ij}]_{1\leq i,j\leq n}\in \cP(\bC^n).
\end{equation} 
Every nonexpansive map $\phi\colon \cP(\bC^n)\to \cP(\bC^n)$ such that $\phi(\cP(\bC^n))$ contains a COSP is either a Wigner symmetry or of the form \eqref{mmnk}.
It is somewhat surprising that the assumption $n \ge 3$ is indispensable. The general form of nonexpansive maps $\phi$ on ${\cal P}(\mathbb{C}^2)$ 
such that $\phi(\cP(\bC^2))$ contains a COSP
 is described in Proposition \ref{dimensiontwo}. 

The next section will be devoted to some auxiliary results. In the third section we will present proofs of the main results. The last section
will be devoted to counterexamples showing the optimality of our main theorems.

\section{Preliminary results}
In what follows we occasionally use the following elementary facts. 
They are well-known and easily verified, so we omit the proofs. 
\begin{lemma}\label{2x2}
A $2\times 2$ projection of rank one is of the form 
$$
\begin{bmatrix} p & z\sqrt{p (1-p)} \\  \overline{z}\sqrt{p (1-p)} & 1-p \end{bmatrix}
$$
for some real number $p$, $0 \leq p \leq1$, and complex number $z$ of modulus one.
\end{lemma}
\begin{lemma}\label{cos}
Let $v, w\in H$ be unit vectors. Then the projections $P_v$ and $P_w$ in $\cP(H)$ onto $\bC v$ and $\bC w$, respectively, satisfy
$$ {\rm tr}\, (P_vP_w) = |\langle v,w\rangle|^2,\ \ \  \|P_v-P_w\| = \sqrt{1-|\langle v,w\rangle|^2}.$$
\end{lemma}

We start with a basic property of nonexpansive maps on $\cP(H)$.

\begin{lemma}\label{inclusion}
Let $\phi : {\cal P} (H) \to {\cal P} (H)$ be a map such that
\begin{equation}\label{xicontr}
\| \phi (P) - \phi (Q) \| \le \| P - Q \|, \ \ \ P,Q \in {\cal P}(H) ,
\end{equation}
and let $n$ be a positive integer. 
Assume that $\{P_1 , \ldots , P_n\} \subset {\cal P}(H)$ is an OSP contained in the image of $\phi$. 
If $Q_1, \ldots, Q_n\in \cP(H)$ satisfy $\phi(Q_j)=P_j$, $j=1, \ldots, n$, then
$\{Q_1 , \ldots , Q_n\} \subset {\cal P}(H)$ is also an OSP. 
Moreover, for every $Q \in {\cal P}(H)$ we have
$$
Q \le Q_1 + \cdots + Q_n \Rightarrow \phi (Q) \le P_1 + \cdots + P_n.
$$
\end{lemma}

\begin{proof}
For $1\leq i< j\leq n$, we know 
$$1=\|P_i-P_j\|=\|\phi(Q_i)-\phi(Q_j)\|\leq \|Q_i-Q_j\|\leq 1 ,$$
which implies that $\{Q_1 , \ldots , Q_n\} \subset {\cal P}(H)$ is an OSP.

Assume that $Q \le Q_1 + \cdots + Q_n$. Then $Q = Q(  Q_1 + \cdots + Q_n )$, and consequently,
\begin{equation}\label{sumone}
1 = {\rm tr}\, Q = {\rm tr}\, (Q Q_1) + \cdots + {\rm tr}\, (Q Q_n) .
\end{equation}
Let $\{ e_1 , \ldots , e_n \}$ be  an orthonormal set of vectors such that  
$P_j = P_{e_j}$, $j=1, \ldots, n$. Then with respect to the direct sum decomposition
$$
H = {\rm span} \, \{ e_1, \ldots , e_n \} \oplus \{ e_1, \ldots , e_n \}^\perp
$$
the rank one projections $P_j$, $j=1, \ldots , n$, have matrix representations
$$
P_j=\phi (Q_j) = \begin{bmatrix}   E_{jj} & 0 \\ 0 & 0 \end{bmatrix}.
$$
Here, $E_{jj}$ stands for the $n \times n$ matrix whose all entries are zero but the $(j,j)$-entry which is equal to $1$.
Let
$$
\phi (Q) = \begin{bmatrix} \begin{bmatrix}  q_{11}&q_{12} & \cdots &q_{1n} \\  q_{21}&q_{22}&\cdots & q_{2n} \\ \vdots & \vdots & \ddots & \vdots \\  q_{n1}&q_{n2}& \cdots & q_{nn}  \end{bmatrix} & A \\ A^\ast & B \end{bmatrix} = 
 \begin{bmatrix}   R & A \\ A^\ast & B \end{bmatrix}
$$
be the corresponding matrix representation of the rank one projection $\phi(Q)$.

Obviously, ${\rm tr}\, (\phi (Q) \phi (Q_j) ) = q_{jj}$. From (\ref{xicontr}) we infer that ${\rm tr}\, (QQ_j) \le {\rm tr}\, (\phi (Q) \phi (Q_j) )$, $j=1, \ldots , n$. Hence, by (\ref{sumone}) we have
$$
1 \le {\rm tr}\, (\phi (Q) \phi (Q_1)) + \cdots + {\rm tr}\, (\phi( Q) \phi( Q_n)) = q_{11} + \cdots + q_{nn} \le {\rm tr} \, \phi (Q) = 1,
$$
and therefore, $q_{11} + \cdots + q_{nn} = 1$. This further implies that $B$ is a positive trace zero operator, and consequently, $B=0$. Using the fact that $\phi (Q)$ is positive we conclude that $A= 0$. Hence, $R$ is a rank one projection and
$$
\phi (Q) = \begin{bmatrix}   R & 0 \\ 0 & 0 \end{bmatrix}\le 
 \begin{bmatrix}   I & 0 \\ 0 & 0 \end{bmatrix} = P_1 + \cdots + P_n .
$$
\end{proof}

From the first part of this lemma, we immediately see the following. 
\begin{corollary}\label{image}
Let $\phi\colon \cP(H)\to \cP(H)$ be a nonexpansive mapping. 
Then the following two conditions are equivalent. 
\begin{itemize}
\item $\phi$ maps some OSP of $\cP(H)$ onto a COSP of $\cP(H)$. 
\item There is a COSP of $\cP(H)$ that is contained in the image of $\phi$.
\end{itemize}
If in addition $\dim H<\infty$, then the above conditions are also equivalent to the condition that 
$\phi$ maps some COSP of $\cP(H)$ onto a COSP of $\cP(H)$. 
\end{corollary}

Let $S^1$ denote the set of all complex numbers of modulus 1. 
We consider a nonexpansive map on $S^1$, that is, a function $g : S^1 \to S^1$ satisfying
$$\lvert g(z_1) - g(z_2)\rvert \leq \lvert z_1-z_2\rvert , \ \ \ z_1, z_2 \in S^1.$$
Clearly, this condition is equivalent to 
\begin{equation}\label{condit2}
{\rm Re}\, (g(z_1)\overline{g(z_2)}) \geq {\rm Re}\, (z_1\overline{z_2}), \ \ \ z_1, z_2 \in S^1.
\end{equation}

We give an easy lemma.
\begin{lemma}\label{sone}
Let $g : S^1 \to S^1$ be nonexpansive.  
Then one of the following holds: 
\begin{itemize}
\item There is $c\in S^1$ such that $g(z)=cz$, $z\in S^1$. 
\item There is $c\in S^1$ such that $g(z)=c\overline{z}$, $z\in S^1$. 
\item The image $g(S^1)\subset S^1$ is compact, connected, and contained in some closed half-circle. 
\end{itemize}
\end{lemma}
\begin{proof}
If $g$ is surjective, then the second item of Theorem \ref{metric} implies that $g$ is a surjective isometry. 
It follows that $g$ satisfies one of the first two options.
Assume that $g$ is not surjective. 
The continuity of $g$ implies that $g(S^1)$ is compact and connected.  
We may assume by rotation that $g(1)=1$ and $g(S^1)=\{e^{i\theta}\,:\, 0\leq \theta \leq \theta_0\}$ with $0\leq \theta_0<2\pi$, without loss of generality. 
For every point $z$ of $S^1$, there is a path in $S^1$ from $1$ to $z$ with length at most $\pi$. 
Therefore, the assumption that $g$ is nonexpansive implies that $\theta_0\leq \pi$.
\end{proof}

The following corollary will be used in the proof of Proposition \ref{diagonal}.
\begin{corollary}\label{homo}
Let $g : S^1 \to S^1$ be nonexpansive. 
If in addition $g$ is a group homomorphism, i.e., if $g(zw)=g(z)g(w)$, $z,w\in S^1$, then one of the following holds: 
\begin{itemize}
\item $g(z)=z$, $z\in S^1$. 
\item $g(z)=\overline{z}$, $z\in S^1$. 
\item $g(S^1)=\{1\}$. 
\end{itemize}
\end{corollary}
\begin{proof}
Trivial in the case of the first two options in Lemma \ref{sone}. 
If the third option in Lemma \ref{sone} holds, then $g(S^1)$ is connected and contained in some closed half-circle. 
This together with the additional assumption clearly shows that $g(S^1)=\{1\}$. 
\end{proof}

Let $g\colon S^1\to S^1$ be a nonexpansive map. 
We define a map $\tau : {\cal P}(\mathbb{C}^2) \to  {\cal P}(\mathbb{C}^2)$ in the following way: We first set
$$
\tau (E_{11}) = E_{11}  \ \ \ {\rm and} \ \ \ \tau (E_{22}) = E_{22}. $$
Let $P$ be an arbitrary rank one projection with $P\not= E_{11}, E_{22}$. By Lemma \ref{2x2},
$$
P = \begin{bmatrix} p & z\sqrt{p (1-p)} \\  \overline{z}\sqrt{p (1-p)} & 1-p \end{bmatrix}
$$
for some real number $p$, $0 < p < 1$, and $z\in S^1$, and for such $P$ we define 
$$
\tau (P) =\begin{bmatrix} p & g(z)\sqrt{p (1-p)} \\   \overline{g(z)} \sqrt{p (1-p)}& 1-p  \end{bmatrix}.
$$
Any such map will be called a standard nonexpansive map of ${\cal P}(\mathbb{C}^2)$. 

We claim that every such map satisfies
$$
\| \tau (P_1) - \tau (P_2) \| \le \| P_1 - P_2 \|, \ \ \ P_1 ,P_2 \in  {\cal P}(\mathbb{C}^2).
$$
Equivalently, we need to show that
\begin{equation}\label{damja}
{\rm tr}\, (\tau (P_1 )  \tau (P_2 ))  \ge {\rm tr}\, ( P_1  P_2 ), \ \ \ P_1 ,P_2 \in  {\cal P}(\mathbb{C}^2).
\end{equation}
It is trivial to verify the above inequality if any of $P_1 ,P_2$ is equal to any of $E_{11}, E_{22}$. If 
$$
P_j = \begin{bmatrix}  p_j & z_j\sqrt{p_j (1-p_j)} \\  \overline{z_j}\sqrt{p_j (1-p_j)} & 1-p_j  \end{bmatrix}, \ \ \ j=1,2,
$$
with $0 < p_j < 1$, $z_j\in S^1$, $j=1,2$, then a straightforward calculation shows that the verification of (\ref{damja}) reduces to (\ref{condit2}).

Note that $\tau$ is a Wigner symmetry if one of the first two options in Lemma \ref{sone} holds. 
If the third option holds, then $\tau$ is noninjective and nonsurjective.

\begin{proposition}\label{dimensiontwo}
Let $\phi : {\cal P}(\mathbb{C}^2) \to  {\cal P}(\mathbb{C}^2)$ be a map satisfying
$$
\| \phi (P) - \phi (Q) \| \le \| P - Q \|, \ \ \ P,Q \in  {\cal P}(\mathbb{C}^2),
$$
and assume that there exists a COSP of $\cP(\bC^2)$ that is contained in the image of $\phi$. 
Then there exist $2 \times 2$ unitary matrices $U,V$ and a standard nonexpansive map $\tau : {\cal P}(\mathbb{C}^2) \to  {\cal P}(\mathbb{C}^2)$ such that
$$
\phi (P)= V \tau (UPU^\ast) V^\ast , \ \ \ P \in {\cal P}(\mathbb{C}^2) .
$$
\end{proposition}

To show this proposition, we first see by Corollary \ref{image} that there is a COSP $\{Q_1 , Q_2\} \in {\cal P}(\mathbb{C}^2)$ such that $\{P_1, P_2\}$ is also a COSP, where $\phi (Q_j) = P_j$, $j=1,2$.
Thus, there exist $2 \times 2$ unitary matrices $U$ and $V$ such that $U^\ast E_{jj} U = Q_j$ and $V^\ast P_j V = E_{jj}$, $j=1,2$. It follows that the map $P \mapsto V^\ast \phi (U^\ast P U)V$, $P\in  {\cal P}(\mathbb{C}^2)$, sends matrices $E_{11}$ and $E_{22}$ to themselves.
Therefore, Proposition \ref{dimensiontwo} reduces to 
\begin{proposition}\label{snem}
Let $\phi : {\cal P}(\mathbb{C}^2) \to  {\cal P}(\mathbb{C}^2)$ be a map satisfying
$$\| \phi (P) - \phi (Q) \| \le \| P - Q \|, \ \ \ P,Q \in  {\cal P}(\mathbb{C}^2),$$
and $\phi(E_{11})=E_{11}$, $\phi(E_{22})=E_{22}$. 
Then $\phi$ is a standard nonexpansive map. 
\end{proposition} 
\begin{proof}
From $\| \phi (P) - \phi (Q) \| \le \| P - Q \|$, $P,Q \in  {\cal P}(\mathbb{C}^2)$, we infer that ${\rm tr}\, (\phi (P) \phi (Q)) \ge {\rm tr}\, (PQ)$. 
Let $P$ be an arbitrary rank one projection with $P\not= E_{11}, E_{22}$. By Lemma \ref{2x2},
\begin{equation}\label{one}
P = \begin{bmatrix} p & z \sqrt{p (1-p)}  \\   \overline{z} \sqrt{p (1-p)} & 1-p \end{bmatrix}
\end{equation}
for some real number $p$, $0 < p < 1$, and some complex number $z \in S^1$. 
Denote
\begin{equation}\label{two}
Q = \phi (P) = \begin{bmatrix}  q & w \sqrt{q (1-q)}  \\   \overline{w} \sqrt{q (1-q)} & 1-q  \end{bmatrix},
\end{equation}
where $0 \le q \le 1$ and $w \in S^1$. We have
$$
q = {\rm tr}\, ( \phi (P) \phi (E_{11})) \ge {\rm tr}\, (P E_{11}) = p,
$$
and similarly, $1-q \ge 1-p$, and consequently, $p=q$. 

In particular,
there exists a function $g : S^1 \to S^1$ such that
$$
\phi (T_u) = T_{g(u)},\ \ \ \text{where}\ \  T_u=\begin{bmatrix}  {1 \over 2} &   {1 \over 2} u  \\   {1 \over 2}   \overline{u} &  {1 \over 2} \end{bmatrix},\ \ \ \text{for every}\ \ u \in S^1.
$$ 
From 
$$
{\rm tr}\, (T_{g(u_1)} T_{g(u_2)} )={\rm tr}\, (\phi (T_{u_1}) \phi (T_{u_2}) ) \ge {\rm tr}\, (T_{u_1}   T_{u_2} ),\ \ \ u_1, u_2\in S^1, 
$$
we deduce that $g\colon S^1\to S^1$ is nonexpansive.

Taking $P$, $Q$ as in (\ref{one}) and (\ref{two}) (hence $p=q$), and applying 
$$
{\rm tr}\, (T_{g(z)}  \phi (P))={\rm tr}\, (\phi (T_z)  \phi (P)) \ge {\rm tr}\, (T_z  P ),
$$ 
we conclude that $w = g (z)$. 
That is, for every 
pair $p, z$, $0 < p < 1$, $z \in S^{1}$, we have
$$
\phi \left( \begin{bmatrix}  p &  z \sqrt{p (1-p)}\\  \overline{z} \sqrt{p (1-p)}  & 1-p \end{bmatrix} \right) 
=
\begin{bmatrix} p &  g(z) \sqrt{p (1-p)} \\  \overline{g(z)} \sqrt{p (1-p)}  & 1-p \end{bmatrix}.
$$
The proof is complete.
\end{proof}

\begin{remark}\label{additional}In the next section we will need the following straightforward consequence. Let $\phi\colon \cP(\bC^n)\to \cP(\bC^n)$ be a nonexpansive map and $P \in  \cP(\bC^n)$ any projection of rank one. Let further $r$ be a positive integer, $1 \le r < n$. We write $P$ in the block form
$$
P = \left[ \begin{matrix} R_1 & N \cr N^\ast & R_2 \end{matrix} \right]
$$
where $R_1$ is an $r \times r$ matrix. 
Since both $R_1$ and $R_2$ are positive and of rank at most one, and because $P$ is a trace one matrix, we have $R_1 = \lambda P_1$ and $R_2 = (1-\lambda) P_2$ for some rank one projections $P_1$ and $P_2$ and for some real $\lambda$, $0 \leq \lambda \leq 1$.

If there exist $r \times r$ projection $Q_1$ and $(n-r) \times (n-r)$ projection $Q_2$ such that
$$\phi \left(  \left[ \begin{matrix} P_1 & 0\cr  0 & 0 \end{matrix} \right] \right) = \left[ \begin{matrix} Q_1 & 0\cr  0 & 0 \end{matrix} \right] \ \ \ {\rm and} \ \ \
\phi \left(  \left[ \begin{matrix} 0 & 0\cr  0 & P_2 \end{matrix} \right] \right) = \left[ \begin{matrix} 0 & 0\cr 0  & Q_2 \end{matrix} \right],
$$
then Proposition \ref{snem} combined with Lemma \ref{inclusion} yields that $\phi (P)$ is a matrix of the form
$$
\phi (P) = \left[ \begin{matrix} \lambda Q_1 & * \cr *  & (1 - \lambda ) Q_2 \end{matrix} \right].
$$
\end{remark}

\section{Proofs of the main results}
We first give a proof of Theorem \ref{Phiinfin} in the finite-dimensional case.
Thus, we assume that $n\geq 3$ and consider a nonexpansive mapping $\phi\colon \cP(\bC^n)\to \cP(\bC^n)$ which maps some COSP onto a COSP.
As in the proof of Proposition \ref{dimensiontwo}, it suffices to prove the following proposition. 

\begin{proposition}\label{diagonal}
Assume $n\geq 3$. 
Let $\phi\colon \cP(\bC^n)\to \cP(\bC^n)$ be a mapping satisfying 
\[
\|\phi(P)-\phi(Q)\|\leq \|P-Q\|,\quad P,Q\in \cP(\bC^n),
\]
and assume that $\phi(E_{jj})=E_{jj}$, $j=1, 2, \ldots, n$.

Then there is an $n \times n$ diagonal unitary matrix $U$ whose $(1,1)$-entry is $1$ such that one of the following three possibilities holds true: 
$$\phi (P)= U P U^\ast , \ \ \ P \in {\cal P}(\mathbb{C}^n) ,$$
$$\phi (P)= U P^t U^\ast , \ \ \ P \in {\cal P}(\mathbb{C}^n) ,$$
where $A^t$ denotes the transpose of a matrix $A$, or 
$$\phi (P)= U \Phi (P) U^\ast , \ \ \ P \in {\cal P}(\mathbb{C}^n) ,$$
where $\Phi$ is given by \eqref{ab}.
\end{proposition}
\begin{proof}
We first consider the case that $n=3$.
Lemma \ref{inclusion} and Proposition \ref{snem} imply that there is a nonexpansive mapping $f_{12}\colon S^1\to S^1$ such that 
\[
\phi\left(
\begin{bmatrix}
p & z \sqrt{p (1-p)}  &0\\
\overline{z} \sqrt{p (1-p)} & 1-p & 0\\
0&0&0
\end{bmatrix}
\right)
=\begin{bmatrix}
p & f_{12}(z) \sqrt{p (1-p)}  &0\\
\overline{f_{12}(z)} \sqrt{p (1-p)} & 1-p & 0\\
0&0&0
\end{bmatrix}
\]
for every $p\in [0, 1]$ and $z\in S^1$.
Similarly,  there are nonexpansive mappings $f_{13}, f_{23}\colon S^1\to S^1$ such that 
\[
\phi\left(
\begin{bmatrix}
p &0& z \sqrt{p (1-p)}  \\ 0&0&0\\
\overline{z} \sqrt{p (1-p)} &0& 1-p
\end{bmatrix}
\right)
=\begin{bmatrix}
p &0&  f_{13}(z) \sqrt{p (1-p)}  \\ 0&0&0\\
\overline{ f_{13}(z)} \sqrt{p (1-p)} &0& 1-p
\end{bmatrix}
\]
and 
\[
\phi\left(
\begin{bmatrix}
0&0&0\\
0&p & z \sqrt{p (1-p)} \\
0&\overline{z} \sqrt{p (1-p)} & 1-p \\
\end{bmatrix}
\right)
=\begin{bmatrix}
0&0&0\\
0&p & f_{23}(z) \sqrt{p (1-p)} \\
0&\overline{f_{23}(z)} \sqrt{p (1-p)} & 1-p\\
\end{bmatrix}
\]
for every $p\in [0, 1]$ and $z\in S^1$. 

Let $p\in [0, 1]$ and $z\in S^1$.
Using Remark \ref{additional} for the pair of orthogonal rank one projections 
$$\begin{bmatrix}
p & z \sqrt{p (1-p)}  &0\\
\overline{z} \sqrt{p (1-p)} & 1-p & 0\\
0&0&0
\end{bmatrix} \ \ \ {\rm and} \ \ \
 E_{33},$$
we see the following: For every $P\in \cP(\bC^3)$ that is of the form $$\begin{bmatrix}
\lambda p & \lambda z\sqrt{p (1-p)} &*\\
\lambda\overline{z} \sqrt{p (1-p)} & \lambda(1-p) & *\\
*&*&1-\lambda
\end{bmatrix}$$
with $0\leq \lambda\leq 1$, the projection
$\phi(P)$ is of the form $$\begin{bmatrix}
\lambda p & \lambda f_{12}(z)\sqrt{p (1-p)} &*\\
\lambda\overline{f_{12}(z)} \sqrt{p (1-p)} & \lambda(1-p) & *\\
*&*&1-\lambda
\end{bmatrix}.$$

We repeat the same arguments with $E_{11}$ and $E_{22}$ instead of $E_{33}$.
It follows that the three mappings $f_{12}, f_{13}, f_{23}$ satisfy
\[
\phi\left(
\begin{bmatrix}
\lambda_1^2 & z\lambda_1\lambda_2 & w\lambda_1\lambda_3\\
\overline{z} \lambda_1\lambda_2 & \lambda_2^2 & \overline{z}w\lambda_2\lambda_3\\
\overline{w} \lambda_1\lambda_3&z\overline{w} \lambda_2\lambda_3&\lambda_3^2
\end{bmatrix}
\right)
=\begin{bmatrix}
\lambda_1^2 & f_{12}(z)\lambda_1\lambda_2 & f_{13}(w)\lambda_1\lambda_3\\
\overline{f_{12}(z)} \lambda_1\lambda_2 & \lambda_2^2 & f_{23}(\overline{z}w)\lambda_2\lambda_3\\
\overline{f_{13}(w)} \lambda_1\lambda_3&\overline{f_{23}(\overline{z}w)} \lambda_2\lambda_3&\lambda_3^2
\end{bmatrix}
\]
for any nonnegative real numbers $\lambda_1,\lambda_2, \lambda_3$ with $\lambda_1^2+\lambda_2^2+\lambda_3^2=1$ and $z, w\in S^1$. 
Since the right-hand side belongs to $\cP(\bC^3)$, we obtain 
$$f_{12}(z)f_{23}(\overline{z}w)=f_{13}(w),\ \ \ z,w\in S^1.$$
By substituting $1$ for $z$, we see that $f_{12}(1)f_{23}(w)=f_{13}(w)$, and by substituting $1$ for $w$, we see that $f_{12}(z)f_{23}(\overline{z})=f_{13}(1)$. 
It follows that $f_{13}(1)\overline{f_{23}(\overline{z})} f_{23}(\overline{z}w)=f_{12}(1)f_{23}(w)$.
Thus we obtain $g(zw)= g(z)g(w)$ for every pair $z,w\in S^1$, where $g(z):=\overline{f_{13}(1)}f_{12}(1)f_{23}(z)$, $z\in S^1$. 
In other words, $g\colon S^1\to S^1$ is a group homomorphism.

Since $f_{23}$ is nonexpansive, so is $g$. 
Therefore, Corollary \ref{homo} implies that one of the three possibilities $g(z)=z$ for every $z\in S^1$, $g(z)=\overline{z}$ for every $z\in S^1$, or $g(S^1)=\{1\}$ holds true. 
It is now straightforward to conclude that 
\begin{itemize}
\item $\phi$ is a Wigner symmetry induced by a diagonal unitary matrix whose $(1,1)$-entry is $1$ if one of the first two options holds, and
\item there is a diagonal unitary matrix $U$ whose $(1,1)$-entry is $1$ such that $\phi([p_{ij}]_{1 \le i,j \le 3})=U[\, |p_{ij}| \, ]_{1 \le i,j \le 3}U^*$ for every $P\in \cP(\bC^3)$ if the third option holds.
\end{itemize}
This completes the proof in the special case when $n=3$.

Now let us consider the general case by induction on $n$.
We assume that the theorem holds for $n-1$, $n\ge 4$.
By Lemma \ref{inclusion} and the induction hypothesis we can assume with no loss of generality that either
\begin{equation}\label{united}
\phi \left( \left[ \begin{matrix} P & 0 \cr 0 & 0 \cr \end{matrix} \right] \right) = \left[ \begin{matrix} P & 0 \cr 0 & 0 \cr \end{matrix} \right], \ \ \ P \in {\cal P}(\mathbb{C}^{n-1}) ,
\end{equation}
or 
\begin{equation}\label{phi1}
\phi \left( \left[ \begin{matrix} P & 0 \cr 0 & 0 \cr \end{matrix} \right] \right) = \left[ \begin{matrix} \Phi (P) & 0 \cr 0 & 0 \cr \end{matrix} \right], \ \ \ P \in {\cal P}(\mathbb{C}^{n-1}) . 
\end{equation}
Here we use the same symbol $\Phi$ for two different maps, one on  $\cP (\mathbb{C}^{n-1})$ and the other one on $\cP (\mathbb{C}^{n})$, both of them sending any projection $P$ of rank one to a rank one projection $Q$ whose entries are the absolute values of the corresponding entries of $P$.

Applying Lemma \ref{inclusion} and the induction hypothesis once again, we see that there is a diagonal $(n-1) \times (n-1)$ unitary matrix $U$ whose $(1,1)$-entry is $1$ such that one of the following three possibilities occurs: 
\begin{equation}\label{migi}
\phi \left( \left[ \begin{matrix} 0 & 0 \cr 0 & P \cr \end{matrix} \right] \right) = \left[ \begin{matrix} 0 & 0 \cr 0 & UPU^\ast \cr \end{matrix} \right]  , \ \ \ P \in {\cal P}(\mathbb{C}^{n-1}) ,
\end{equation}
$$\phi \left( \left[ \begin{matrix} 0 & 0 \cr 0 & P \cr \end{matrix} \right] \right) = \left[ \begin{matrix} 0 & 0 \cr 0 & UP^tU^\ast \cr \end{matrix} \right]  , \ \ \ P \in {\cal P}(\mathbb{C}^{n-1}) ,$$
or 
\begin{equation}\label{phi2}
\phi \left( \left[ \begin{matrix} 0 & 0 \cr 0 & P \cr \end{matrix} \right] \right) = \left[ \begin{matrix} 0 & 0 \cr 0 & U\Phi(P)U^\ast \cr \end{matrix} \right]  , \ \ \ P \in {\cal P}(\mathbb{C}^{n-1}) .
\end{equation}
By considering projections of the form 
$$\left[ \begin{matrix} 0 & 0 &0 \cr 0 & P&0 \cr 0&0&0 \cr \end{matrix} \right]$$
for $P\in  {\cal P}(\mathbb{C}^{n-2})$,
we see that \eqref{united} implies \eqref{migi}, and that \eqref{phi1} implies \eqref{phi2}.
Moreover, in both cases all diagonal entries but the $(n-1, n-1)$-entry of $U$ are $1$.
Thus, by considering the mapping  
$$P\mapsto \left[ \begin{matrix} 1 & 0 \cr 0 & U^* \cr \end{matrix} \right]\phi(P)\left[ \begin{matrix} 1 & 0 \cr 0 & U \cr \end{matrix} \right] $$
instead of $\phi$, we may additionally assume $U=I$.

Let $P =[p_{ij}]_{1 \le i,j \le n}\in  \cP (\mathbb{C}^{n})$ be any rank one projection. 
Then one may take a projection $Q\in \cP (\mathbb{C}^{n-1})$ such that 
$$P=\left[ \begin{matrix} (1-p_{nn}) Q&* \cr *&p_{nn} \cr \end{matrix} \right] .
$$
If \eqref{united} holds, then we have 
$$
\phi \left( \left[ \begin{matrix} Q & 0 \cr 0 & 0 \cr \end{matrix} \right] \right) = \left[ \begin{matrix} Q & 0 \cr 0 & 0 \cr \end{matrix} \right].
$$
Thus Remark \ref{additional} together with the assumption $\phi(E_{nn})=E_{nn}$ shows that 
\begin{equation}\label{ichi}
\phi (P) = 
\phi\left(\left[ \begin{matrix} (1-p_{nn}) Q&* \cr *&p_{nn} \cr \end{matrix} \right]\right)
= \left[ \begin{matrix}  [p_{ij} ]_{1 \le i,j \le n-1} & * \cr * & p_{nn} \cr \end{matrix} \right] .
\end{equation}
Similarly, we obtain
\begin{equation}\label{ni}
\phi (P) = \left[ \begin{matrix}  p_{11} & * \cr * & [p_{ij} ]_{2 \le i,j \le n} \cr \end{matrix} \right]
\end{equation}
because $\phi(E_{11})=E_{11}$ and \eqref{migi} holds with $U=I$. 
Since $\phi(P)$ is a projection of rank one, \eqref{ichi} and \eqref{ni} imply that $\phi(P)=P$ holds for every $P \in \cP (\mathbb{C}^{n})$ with $P\not\le E_{11}+E_{nn}$. 
By the continuity of $\phi$, we see that $\phi(P)=P$ holds for every $P\in \cP (\mathbb{C}^{n})$.

In a parallel manner, if \eqref{phi1} holds, then we have 
\[
\phi (P) = \left[ \begin{matrix}  [\, |p_{ij}|\, ]_{1 \le i,j \le n-1} & * \cr * & p_{nn} \cr \end{matrix} \right] =\left[ \begin{matrix}  p_{11} & * \cr * & [\, |p_{ij}|\, ]_{2 \le i,j \le n} \cr \end{matrix} \right]
\]
for every $P =[p_{ij}]_{1 \le i,j \le n}\in  \cP (\mathbb{C}^{n})$.
This leads to $\phi(P)=\Phi(P)$ for every $P\in \cP(\bC^n)$.
\end{proof}

\begin{remark}\label{antiunitary}
The following facts can be verified easily. 

The unitary matrix $U$ in the statement of Proposition \ref{diagonal} is unique.
It cannot happen that two different possibilities in the statement of Proposition \ref{diagonal} are fulfilled simultaneously.

Let $U$ be a diagonal $n \times n$ unitary matrix. 
Define $V\colon \bC^n \to \bC^n$ by $Vx = U\overline{x}$, $x\in \bC^n$, where $\overline{x}=(\overline{x_i})_{1\leq i\leq n}$ for $x=(x_i)_{1\leq i\leq n}\in \bC^n$. 
Then $V$ is an antiunitary operator satisfying 
$$U P^t U^\ast =VPV^\ast, \ \ \ P \in {\cal P}(\mathbb{C}^n).$$
\end{remark}

Let  $\phi : {\cal P} (H) \to {\cal P} (H)$ be a nonexpansive map.
Assume that there is an OSP $\{ Q_{\alpha} \, : \, \alpha \in {\cal J} \} \subset {\cal P}(H)$ such that $\{ P_{\alpha} \, : \, \alpha \in {\cal J} \} \subset {\cal P}(H)$ is a COSP, where $P_{\alpha}=\phi(Q_\alpha)$, $\alpha\in {\cal J}$. 
Corollary \ref{image} implies that this assumption is fulfilled if and only if some COSP
is contained in the image of $\phi$. 
Note that the latter condition is much weaker than the surjectivity of $\phi$.

We take a unit vector $e_\alpha\in P_\alpha H$ for each $\alpha\in {\cal J}$.  
Then $\{e_\alpha\,:\, \alpha
\in {\cal J}\}$ is an orthonormal basis of $H$. A unitary operator $V : H \to H$ is said to be diagonal (with respect to the orthonormal basis $\{e_\alpha\,:\, \alpha
\in {\cal J}\}$) if there exist complex numbers  $\{z_\alpha\,:\, \alpha
\in {\cal J}\} \subset S^1$ such that $Ve_\alpha = z_\alpha e_\alpha$ for every $\alpha \in {\cal J}$.
For each projection $P\in \cP(H)$ onto $\bC(\sum_{\alpha\in {\cal J}}c_\alpha e_\alpha)$, $\Phi(P)$ stands for the projection onto $\bC(\sum_{\alpha\in {\cal J}}|c_\alpha|e_\alpha)$.
For a closed subspace $K\subset H$, we identify $\cP(K)$ with a subset of $\cP(H)$ in a natural manner.

\begin{proposition}\label{Phicor}
Let $\dim H\geq 3$ and let
$\phi : {\cal P} (H) \to {\cal P} (H)$ be a map such that 
$$
\| \phi (P) - \phi (Q) \| \le \| P - Q \|
$$
for every pair $P,Q \in {\cal P}(H)$. 
Let $\{ Q_{\alpha} \, : \, \alpha \in {\cal J} \} \subset {\cal P}(H)$ be an OSP such that $\{ P_{\alpha} \, : \, \alpha \in {\cal J} \} \subset {\cal P}(H)$ is a COSP, where $P_{\alpha}=\phi(Q_\alpha)$, $\alpha\in {\cal J}$. 
Let $K:=(\bigvee_{\alpha\in {\cal J}} Q_\alpha)H$ and let $\Phi$ be as above.
Then one of the following holds: 
\begin{itemize}
\item $K=H$ and $\phi$ is a Wigner symmetry.
\item There are a bijective linear isometry $U\colon K\to H$ and a diagonal unitary operator $V : H \to H$ such that $\phi(P)= V \Phi(UPU^*) V^\ast$ for all $P\in \cP(K)$. 
\end{itemize}
\end{proposition}
\begin{proof} 
Let $\{ f_\alpha \, : \, \alpha \in {\cal J} \}$ be an orthonormal basis of $K$ such that the image of each $Q_\alpha$ is spanned by $f_\alpha$ and $\{ e_\alpha \, : \, \alpha \in {\cal J} \}$ an orthonormal basis of $H$ such that the image of each $P_\alpha$ is spanned by $e_\alpha$. 
We first define the linear bijective isometry $U : K \to H$ by $U f_\alpha = e_\alpha$, $\alpha \in {\cal J}$.
Observe that $UQ_\alpha U^*= P_{\alpha}=\phi(Q_\alpha)$ holds for every $\alpha\in {\cal J}$.  
We choose and fix distinct elements $\alpha_1, \alpha_2, \alpha_3 \in {\cal J}$. Using Proposition \ref{diagonal} and Lemma \ref{inclusion} (see also Remark \ref{antiunitary}) we see that for every finite subset ${\cal K} \subset {\cal J }$ satisfying $\alpha_1, \alpha_2, \alpha_3 \in {\cal K}$ either 
\begin{itemize}
\item there exists a uniquely determined unitary or antiunitary operator $V_{\cal K} : {\rm span}\, \{ e_\alpha \, : \, \alpha \in {\cal K} \}  \to {\rm span}\, \{ e_\alpha \, : \, \alpha \in {\cal K} \}$ such that $V_{\cal K} {e_{\alpha_1}} = e_{\alpha_1}$ and $\phi (Q) = V_{\cal K} UQU^\ast V_{\cal K}^\ast$ for every $Q\in \cP(K)$ with
$Q \le \sum_{\alpha \in {\cal K}} Q_\alpha$; or
\item  there exists a uniquely determined 
unitary operator $V_{\cal K} : {\rm span}\, \{ e_\alpha \, : \, \alpha \in {\cal K} \}  \to {\rm span}\, \{ e_\alpha \, : \, \alpha \in {\cal K} \}$ such that $V_{\cal K} {e_{\alpha_1}} = e_{\alpha_1}$ and $\phi (Q) = V_{\cal K} \Phi (UQU^\ast) V_{\cal K}^\ast$ for every $Q\in \cP(K)$ with
$Q \le \sum_{\alpha \in {\cal K}} Q_\alpha$.
\end{itemize}
Obviously, if ${\cal K}_1 , {\cal K}_2 \subset {\cal J}$ are finite subsets satisfying $\alpha_1, \alpha_2, \alpha_3\in {\cal K}_1 \subset {\cal K}_2$ then either we have for both sets ${\cal K}_1$ and ${\cal K}_2$ the first possibility above, or we have for both sets the second possibility. In both cases
$V_{{\cal K}_1}$ is the restriction of the operator $V_{{\cal K}_2}$.
It follows that either
\begin{itemize}
\item there exists a bijective linear or conjugate-linear isometry $W :  K \to H$ such that $\phi (P) = WPW^\ast$ for every $P \in {\cal P}(K)$ that is dominated by a sum of finitely many projections in $\{Q_{\alpha}\,:\, \alpha\in {\cal J}\}$; or
\item  there exists a diagonal unitary operator $V: H \to H$ such that $\phi (P) = V \Phi (UPU^\ast) V^\ast$ for every $P \in {\cal P}(K)$ that is dominated by a sum of finitely many projections in $\{Q_{\alpha}\,:\, \alpha\in {\cal J}\}$.
\end{itemize}
By the continuity of $\phi$ we see that one of the above possibilities holds for every $P \in {\cal P}(K)$.

Assume that the first option holds and $K\neq H$. 
Let $v\in K$ and $w\in K^{\perp}$ be unit vectors and $0<c\leq 1$ a real number.
Then we may take a unique projection $Q\in \cP(K)$ such that $\phi(P_{cv+\sqrt{1-c^2}w}) = \phi(Q)$.
For every unit vector $u\in K$ with $Qu=0$, the projection $\phi(P_{u})=WP_u W^*$ is orthogonal to $WQW^*=\phi(Q)=\phi(P_{cv+\sqrt{1-c^2}w})$.
Since $\phi$ is nonexpansive, we see that $P_{u}$ is orthogonal to $P_{cv+\sqrt{1-c^2}w}$. 
Hence $u$ is orthogonal to $v$. Because this is true for every $u \in K$ satisfying $Qu=0$ we conclude
that $Q=P_v$, thus $\phi(P_{cv+\sqrt{1-c^2}w}) = \phi(P_v)$.

Since $P_{cv+\sqrt{1-c^2}w}$ converges to $P_w$ when $c$ tends to zero and since $v\in K$ and $w\in K^{\perp}$ are arbitrary unit vectors 
we see that $\phi$ is not continuous at every point of $\cP(K^\perp)$. 
This contradicts our assumption, and we obtain $K=H$. 
Thus we have shown that $\phi$ is a Wigner symmetry.
\end{proof}

\begin{proof}[Proof of Theorem \ref{Phiinfin}]
In the formulation of Theorem \ref{Phiinfin} the orthonormal basis $\{ e_\alpha\, : \, \alpha \in {\cal J} \}$ is fixed in advance. In Proposition \ref{Phicor} the orthonormal basis $\{ e_\alpha\, : \, \alpha \in {\cal J} \}$ was determined using a given CONS. Thus, in order to show that Theorem \ref{Phiinfin} follows from Proposition \ref{Phicor} we only need to check that if $\{ e_\alpha\, : \, \alpha \in {\cal J} \}$ and $\{ e'_\alpha\, : \, \beta \in {\cal J}' \}$ are two orthonormal bases of $H$ and $\Phi$ and $\Phi'$ are the corresponding nonexpansive maps, then there exists a unitary operator $W : H \to H$ such that $W\Phi(P) W^\ast = \Phi' (WPW^\ast)$, $P \in {\cal P}(H)$. There exists a bijection $\tau : {\cal J} \to {\cal J}'$. It is straightforward to verify that the unitary operator $W$ given by $We_{\alpha} = e'_{\tau (\alpha)}$, $\alpha \in {\cal J}$, has the desired property.
\end{proof}

\begin{proof}[Proof of Theorem \ref{plasticity}] 
If $\dim H=2$ then the second item of Corollary \ref{known} leads to the desired conclusion.
In what follows we assume that $\dim H\geq 3$. 
The surjectivity of $\phi$ together with Corollary \ref{image} implies that $\phi$ satisfies the assumption of Proposition \ref{Phicor}.

Assume that $\phi$ is of the form as in the second option of the statement of Proposition \ref{Phicor}.
The goal of this proof is to show that $\phi$ is never surjective. 
From now on we use the same symbols as in Proposition \ref{Phicor}. 
After replacing the map $\phi$ by $P \mapsto V^\ast \phi (P) V$, $P \in \cP (H)$, we may assume that $V=I$.
Fix three indices $\alpha_1, \alpha_2, \alpha_3\in {\cal J}$.
Fix any complex numbers $c_1, c_2, c_3\in \bC$ of modulus one such that $0\in \bC$ is contained in the interior of the convex hull of $\{c_1, c_2, c_3\}$. For example, we may choose the cube roots of $1$.

We can find  $\varepsilon >0$ such that $D(0,\varepsilon) = \{ z \in \mathbb{C} \, : \, |z| < \varepsilon \}$ is contained in the convex hull of $\{c_1, c_2, c_3\}$. Then there exists $\delta >0$ such that $D(0, \varepsilon /2 )$ is 
contained in the convex hull of $\{c'_1, c'_2, c'_3\}$ for all complex numbers $c'_1, c'_2, c'_3$ satisfying $| c'_j - c_j | < \delta$.

Let $P\in \cP(H)$ denote the projection onto $\bC(c_1e_{\alpha_1}+c_2e_{\alpha_2}+c_3e_{\alpha_3})$. 
We can find $\eta>0$ such that if
$P'\in \cP(H)$ satisfies $\|P-P'\|<\eta$, then we may take $c'_1, c'_2, c'_3 \in \bC$ and $v\in \{e_{\alpha_1} , e_{\alpha_2}, e_{\alpha_3} \}^\perp$ such that $P'$ is the projection onto $\bC(c_1'e_{\alpha_1}+c_2'e_{\alpha_2}+c_3'e_{\alpha_3}+v)$ and
$$
| c'_j - c_j | < \delta \ \ \ {\rm and} \ \ \ \|v \| < {\varepsilon \over 2}.
$$

Let $P'\in \cP(H)$ be a projection satisfying $\|P-P'\|<\eta$ and let $c'_1, c'_2, c'_3 \in \bC$ and $v\in \{e_{\alpha_1} , e_{\alpha_2}, e_{\alpha_3} \}^\perp$ be as above.
Let further $\{a_\alpha\,:\, \alpha \in {\cal J}\setminus\{\alpha_1,\alpha_2, \alpha_3\}\}$ be any family of nonnegative real numbers such that $\sum_{\alpha\in {\cal J}\setminus\{\alpha_1,\alpha_2, \alpha_3\}} a_{\alpha}^2< 1$. 
Then
$$\left| \left\langle v, \sum_{\alpha\in {\cal J}\setminus\{\alpha_1,\alpha_2, \alpha_3\}} a_{\alpha}e_\alpha \right\rangle \right| < {\varepsilon \over 2} .$$
Thus, we can find $a_{\alpha_1}, a_{\alpha_2}, a_{\alpha_3}>0$ such that $a_{\alpha_1}+ a_{\alpha_2}+ a_{\alpha_3}=1$ and 
\[
a_{\alpha_1}c_1'+a_{\alpha_2}c_2'+a_{\alpha_3}c_3'=- \left\langle v, \sum_{\alpha\in {\cal J}\setminus\{\alpha_1,\alpha_2, \alpha_3\}} a_{\alpha}e_\alpha \right\rangle
\]
hold. 
It follows that the projection $E$ onto $\bC\left(\sum_{\alpha\in {\cal J}} a_{\alpha}e_\alpha\right)$ is orthogonal to $P'$. 
Observe that the following holds: 
For any family of complex numbers $\{b_\alpha\,:\, \alpha\in {\cal J}\}$, with $|b_\alpha|=a_\alpha$ for all $\alpha\in {\cal J}$, the projection $F$ onto $\bC\left(\sum_{\alpha\in {\cal J}} b_{\alpha}f_\alpha\right)$ satisfies $\phi(F)=E$. 
By the assumption that $\phi$ is nonexpansive, we see that such an $F$ is orthogonal to any projection in $\phi^{-1}(P')$. 

The arbitrariness of $\{a_\alpha\,:\, \alpha \in {\cal J}\setminus\{\alpha_1,\alpha_2, \alpha_3\}\}$ and $\{b_\alpha\,:\, \alpha \in {\cal J}\}$ ensures that the one-dimensional image of any projection in $\phi^{-1}(P')$ is actually orthogonal to $K$. 
It follows that $\phi^{-1}(\{P'\in \cP(H)\,:\, \|P-P'\|<\eta\})\subset \cP(K^{\perp})$. 
The continuity of $\phi$ implies that the left-hand side is open in $\cP(H)$. 
However, the interior of $\cP(K^{\perp})$ is empty in $\cP(H)$. 
Therefore, we obtain $\phi^{-1}(\{P'\in \cP(H)\,:\, \|P-P'\|<\eta\})=\emptyset$. 
In particular, $\phi$ is not surjective.
\end{proof}

\begin{proof}[Proof of Theorem \ref{plasticity3}]
Obviously, $\phi$ is bijective.  By Theorem \ref{plasticity}, the inverse of $\phi$ is a Wigner symmetry. Therefore, $\phi$ is a Wigner symmetry.
\end{proof}

\section{Optimality of results}

The surjectivity assumption is essential in Theorem \ref{plasticity3}.  
\begin{example}\label{ce1}
Assume $\dim H=\infty$. 
Then there is a noncontractive map $\phi :  {\cal P}(H) \to {\cal P}(H)$ such that no linear or conjugate-linear isometry $U : H \to H$ satisfies $\phi (P) = UPU^\ast$ for every $P \in {\cal P}(H)$.
\end{example}
\begin{proof}
Note that an infinite-dimensional Hilbert space $H$ can be identified with the orthogonal direct sum of two copies of $H$, $H \equiv H \oplus H$. Then a map $\phi$ on ${\cal P}(H)$ can be considered as a map from ${\cal P}(H)$ to ${\cal P}(H \oplus H)$. We represent operators in  ${\cal P}(H \oplus H)$ with $2 \times 2$ operator matrices. Choose and fix $\emptyset \neq \mathcal{A} \subsetneq {\cal P}(H)$. It is trivial to verify that the map $\phi :  {\cal P}(H) \to {\cal P}(H \oplus H)$ given by
$$
\phi (P) =\begin{bmatrix} P & 0 \\ 0 & 0  \end{bmatrix}, \ \ \ P \in {\cal P}(H) \setminus \mathcal{A},
$$
and
$$
\phi (P) = \begin{bmatrix} 0 & 0 \\ 0 & P \end{bmatrix}, \ \ \ P \in \mathcal{A}
$$
is noncontractive but there is no linear or conjugate-linear isometry $U : H \to H \oplus H$ such that
$\phi (P) = UPU^\ast$ for every $P \in {\cal P}(H)$.
\end{proof}

We next give an example of an injective nonexpansive map. 
\begin{example}\label{ef}
Let $H$ be an infinite-dimensional separable Hilbert space.
There is an injective nonexpansive map $\phi : {\cal P} (H) \to {\cal P} (H)$ such that no linear or conjugate-linear isometry $U : H \to H$ satisfies $\phi (P) = UPU^\ast$ for every $P \in {\cal P}(H)$.
\end{example}
\begin{proof}
We will identify ${\cal P}(H)$ with the projective space $\mathbb{P} (H) = \{ [x] \, : \, x \in H \setminus \{ 0 \} \}$ over the Hilbert space $H$. Here, $[x]$ denotes the one-dimensional subspace of $H$ spanned by $x$. Of course, we identify a projection of rank one with its image. 
By Lemma \ref{cos}, we need to find an injective map $\phi : \mathbb{P} (H) \to \mathbb{P} (H)$ with the property that for any unit vectors $x,y \in H$ and for any unit vectors $u \in \phi ([x])$ and $v \in \phi ([y])$ we have
\begin{equation}\label{ip}
| \langle u,v \rangle | \ge | \langle x,y \rangle |,
\end{equation}
but there is no linear or conjugate-linear isometry $U : H \to H$ satisfying
$\phi ([x]) = [Ux]$ for every unit vector $x\in H$.

We choose and fix a sequence $(P_n) \in {\cal P}(H)$ with the property that the set $\{ P_n \, : \, n \in \mathbb{N} \}$ is dense in ${\cal P}(H)$. Let $(x_n)  \subset H$  be a sequence of unit vectors such that the image of $P_n$ is the linear span of $x_n$, $n=1,2,\ldots$. We next choose and fix an orthonormal set
$\{ e_1 , f_1 , e_2 , f_2 , \ldots \}$ in $H$. For any unit vector $x\in H$ we set $| \langle x, x_n \rangle | = t_n$, $n=1,2,\ldots$, and define
$$
\phi ([x]) = \left[ \sum_{n=1}^\infty { t_n \over 2^{n/2}} e_n +  \sum_{n=1}^\infty { \sqrt{ 1 - t_{n}^2} \over 2^{n/2}} f_n \right].
$$
It is trivial to see that $\phi$ is well-defined. 

Let $x,y \in H$ be any unit vectors. Set $| \langle x, x_n \rangle | = t_n$ and  $| \langle y, x_n \rangle | = s_n$, $n=1,2,\ldots$. Denote
$$
u = \sum_{n=1}^\infty { t_n \over 2^{n/2}} e_n +  \sum_{n=1}^\infty { \sqrt{ 1 - t_{n}^2} \over 2^{n/2}} f_n
$$
and
$$
v = \sum_{n=1}^\infty { s_n \over 2^{n/2}} e_n +  \sum_{n=1}^\infty { \sqrt{ 1 - s_{n}^2} \over 2^{n/2}} f_n .
$$
A straightforward computation shows that $u$ and $v$ are unit vectors. Clearly, $[u]= \phi ([x])$ and $[v] =\phi ([y])$. We need to show \eqref{ip}. We have
$$
| \langle u,v \rangle | = \sum_{n=1}^\infty {1 \over 2^n} ( t_n s_n + \sqrt{ 1 - t_{n}^2} \, \sqrt{ 1 - s_{n}^2} )
$$
and therefore it is enough to check that
\begin{equation}\label{innerprod}
 | \langle x,y \rangle | \le  t_n s_n + \sqrt{ 1 - t_{n}^2} \, \sqrt{ 1 - s_{n}^2}, \ \ \ n=1,2,\ldots.
\end{equation}
Let $z\in H$ be an arbitrary unit vector. Denote $w_1 = x - \langle x,z \rangle z$ and $w_2  = y - \langle y,z \rangle z$. Then
\[
\begin{split}
| \langle x,y \rangle | &= |  \langle x,z \rangle \overline{  \langle y,z \rangle} +  \langle w_1,w_2 \rangle |\\
&\le  |  \langle x,z \rangle | \,  |  \langle y,z \rangle | + \| w_1 \| \, \| w_2 \|\\
&= |  \langle x,z \rangle | \,  |  \langle y,z \rangle | + \sqrt{ 1 -  |  \langle x,z \rangle |^2} \, \sqrt{ 1 -  |  \langle y,z \rangle |^2}.
\end{split}
\]
Substituting $z= x_n$ we get the desired inequality (\ref{innerprod}).
Moreover, if $[x]\neq [y]$, then the assumption that $\{ P_n \, : \, n \in \mathbb{N} \}$ is dense in ${\cal P}(H)$ implies that $\{x_n, x, y\}$ is linearly independent for some $n$.
For such an $n$ and $z=x_n$, we see that $\{w_1, w_2\}$ is linearly independent. Thus we have strict inequality in (\ref{innerprod}) and \eqref{ip}. 
It follows that there is no linear or conjugate-linear isometry $U : H \to H$ satisfying
$\phi ([x]) = [Ux]$ for every unit vector $x\in H$.

It remains to show that $\phi$ is injective. Here we prefer to work with ${\cal P}(H)$ rather than with $\mathbb{P} (H)$. Thus, let $P$ and $Q$ be rank one projections such that $\phi (P) = \phi (Q)$. Then 
$$
{\rm tr}\, (P P_n) = {\rm tr}\, (Q P_n), \ \ \ n=1,2,\ldots ,
$$
and since $\{ P_n \, : \, n=1,2,\ldots \}$ is dense in ${\cal P}(H)$ we have 
$$
{\rm tr}\, (P R) = {\rm tr}\, (Q R)
$$
for every $R \in {\cal P}(H)$. Equivalently, $\| P -R \| = \| Q -R \|$ for every $R\in {\cal P}(H)$, and therefore, $P=Q$, as desired.
\end{proof}

Our last example will show that when the second possibility in Proposition \ref{Phicor} occurs, it may happen that $K$ is a proper subspace of $H$. Indeed, let $H$ be an infinite-dimensional Hilbert space and $\{ f_\alpha \, : \, \alpha \in {\cal J} \}$ an orthonormal set of vectors such that $K = \overline{ {\rm span}\, \{ f_\alpha \, : \, \alpha \in {\cal J} \} }$ is a proper subspace of $H$. Let further $\{ e_\alpha \, : \, \alpha \in {\cal J} \}$ be an orthonormal basis of $H$ and $U : K \to H$ a bijective linear isometry defined by $U f_\alpha = e_\alpha$, $\alpha \in {\cal J}$.
For each projection $P\in \cP(H)$ onto $\bC(\sum_{\alpha\in {\cal J}}c_\alpha e_\alpha)$, $\Phi(P)$ stands for the projection onto $\bC(\sum_{\alpha\in {\cal J}}|c_\alpha|e_\alpha)$. 
Fix an index $\alpha_0\in {\cal J}$, and for each unit vector $w=\sum_{\alpha\in {\cal J}} b_{\alpha} f_\alpha + v$, $v\in K^{\perp}$, let $\phi(P_w)$ be the projection onto $\bC\left(\sum_{\alpha\in {\cal J}\setminus \{ \alpha_0 \}} |b_{\alpha}|e_\alpha + \sqrt{\|v\|^2+|b_{\alpha_0}|^2}e_{\alpha_0}\right)$. 
We denote $P_{f_{\alpha}} = Q_\alpha$ and $P_{e_{\alpha}} = P_\alpha$, $\alpha \in {\cal J}$. Then clearly, $\phi (Q_\alpha) = P_\alpha$, $\alpha \in {\cal J}$, and $\phi (P) = \Phi (UPU^\ast)$, $P \in {\cal P}(K)$. It remains to show that $\phi$ is a nonexpansive map, that is, we have to verify that
$$
{\rm tr} \, ( \phi (P_w) \phi (P_z)) \ge {\rm tr}\, (P_w P_z)
$$
for every pair of unit vectors $w,z \in H$. Let
$$w=\sum_{\alpha\in {\cal J}} b_{\alpha} f_\alpha + v \ \ \ {\rm and} \ \ \
z=\sum_{\alpha\in {\cal J}} c_{\alpha} f_\alpha + u,
$$
$v,u \in K^\perp$, be such vectors. Applying the Cauchy-Schwarz inequality in the two-dimensional Euclidean space we get
$$
| b_{\alpha_0} | \, | c_{\alpha_0} | + \| v \| \, \| u \| \le \sqrt{ \| v \|^2 + |b_{\alpha_0}|^2 } \  \sqrt{ \| u \|^2 + |c_{\alpha_0}|^2 }.
$$
Therefore,
\[
\begin{split}
\sqrt{ {\rm tr}\, ( P_w P_z ) } &= | \langle w,z \rangle | = \left| \sum_{\alpha \in {\cal J}} b_\alpha \overline{ c_\alpha} + \langle v,u \rangle \right|
\le \sum_{\alpha \in {\cal J}} | b_\alpha | \, | c_\alpha| +  | \langle v,u \rangle |\\
&\le  | b_{\alpha_0} | \, | c_{\alpha_0} | + \| v \| \, \| u \|  +  \sum_{\alpha \in {\cal J}\setminus \{ \alpha_0 \} } | b_\alpha | \, | c_\alpha| \\
&\le \sqrt{ \| v \|^2 + |b_{\alpha_0}|^2 } \  \sqrt{ \| u \|^2 + |c_{\alpha_0}|^2 } \ \ + \sum_{\alpha \in {\cal J}\setminus \{ \alpha_0 \} } | b_\alpha | \, | c_\alpha| \\
&=\sqrt{ {\rm tr}\, ( \phi( P_w ) \phi( P_z )) } .
\end{split}
\]

\end{document}